\documentclass[12pt,draftclsnofoot,onecolumn,journal,letterpaper]{IEEEtran}

\usepackage{cite}
\usepackage{array}
\usepackage{bm}
\usepackage{booktabs}
\usepackage{yfonts}
\usepackage{caption}
\usepackage{subcaption}
\usepackage{amssymb}
\usepackage{multirow}
\usepackage{color}
\usepackage{graphicx}
\usepackage{algorithm}
\usepackage{algorithmic}
\usepackage{amsmath}
\usepackage{amsthm}
\usepackage{amssymb}
\usepackage{mathrsfs}

\newtheorem{theorem}{Theorem}

\begin{document}

\title{A Novel and Efficient Vector Quantization Based CPRI Compression Algorithm}

\author{\IEEEauthorblockN{Hongbo~Si, Boon~Loong~Ng, Md.~Saifur~Rahman, and Jianzhong~(Charlie)~Zhang}
\thanks{
The material in this paper was submitted in part to the 2015 IEEE Global Communications Conference (Globecom 2015), San Diego, CA, USA, Dec. 2015.}
\thanks{H.~Si, B.~L.~Ng, M.~S.~Rahman, and J.~(C.)~Zhang are with Samsung Research America - Dallas, 1301 E Lookout Dr, Richardson, TX 75082.
Email: \{hongbo.si, b.ng, md.rahman, jianzhong.z\}@samsung.com.}
}



\maketitle


\begin{abstract}

The future wireless network, such as Centralized Radio Access Network (C-RAN), will need to deliver data rate about $100$ to $1000$ times the current 4G technology. For C-RAN based network architecture, there is a pressing need for tremendous enhancement of the effective data rate of the Common Public Radio Interface (CPRI). Compression of CPRI data is one of the potential enhancements. In this paper, we introduce a vector quantization based compression algorithm for CPRI links, utilizing Lloyd algorithm. Methods to vectorize the I/Q samples and enhanced initialization of Lloyd algorithm for codebook training are investigated for improved performance. Multi-stage vector quantization and unequally protected multi-group quantization are considered to reduce codebook search complexity and codebook size. Simulation results show that our solution can achieve compression of $4$ times for uplink and $4.5$ times for downlink, within $2\%$ Error Vector Magnitude (EVM) distortion. Remarkably, vector quantization codebook proves to be quite robust against data modulation mismatch, fading, signal-to-noise ratio (SNR) and Doppler spread.

\end{abstract}



\section{Introduction}

The amount of wireless IP data traffic is projected to grow by well over $100$ times within a decade (from under $3$ exabytes in 2010 to more than $500$ exabytes by 2020) \cite{cisco2014vni}. To address such wireless data traffic demand, there has been increasing effort to define the 5G network in recent years. It is widely recognized that the 5G network will be required to deliver data rate about $100$ to $1000$ times the current 4G technology, utilizing radical increase in wireless bandwidths at very high frequencies, extreme network densitification, and massive number of antennas \cite{andrews2014will}.

Distributed base station architecture and Cloud Radio Access Network (C-RAN) will continue to be an important network architecture well into the future \cite{cmcc2011cran}. Therefore, there is a pressing need to drastically enhance the data rate of the Common Public Radio Interface (CPRI) (see Fig.~\ref{fig:cpri}), which is the industry standard for interface between the Baseband Units (BBU) and the Remote Radio Units (RRU).

One way to address the significant increase in the CPRI data rate is to deploy more links (typically fibers) connecting the BBUs and the RRUs, but such deployment would incur extraordinary high cost. An alternative method, which can be much more cost effective, is to employ data compression over CPRI links. It is impossible to utilize only CPRI link compression to meet the CPRI link data rate requirement. Nevertheless, CPRI link compression can greatly reduce the required cost when employed in conjunction with new links deployment. Rate reduction between the BBU and the RRU can also be achieved by moving some of the functions traditionally performed at the BBU to the RRU, but this requires a significant change to the existing distributed base station or C-RAN architecture \cite{radisys2014cran}. This paper focuses on CPRI link data compression which can be employed with minimal change to the current distributed base station or C-RAN architecture.

\begin{figure}[t!]
\includegraphics[scale=0.6]{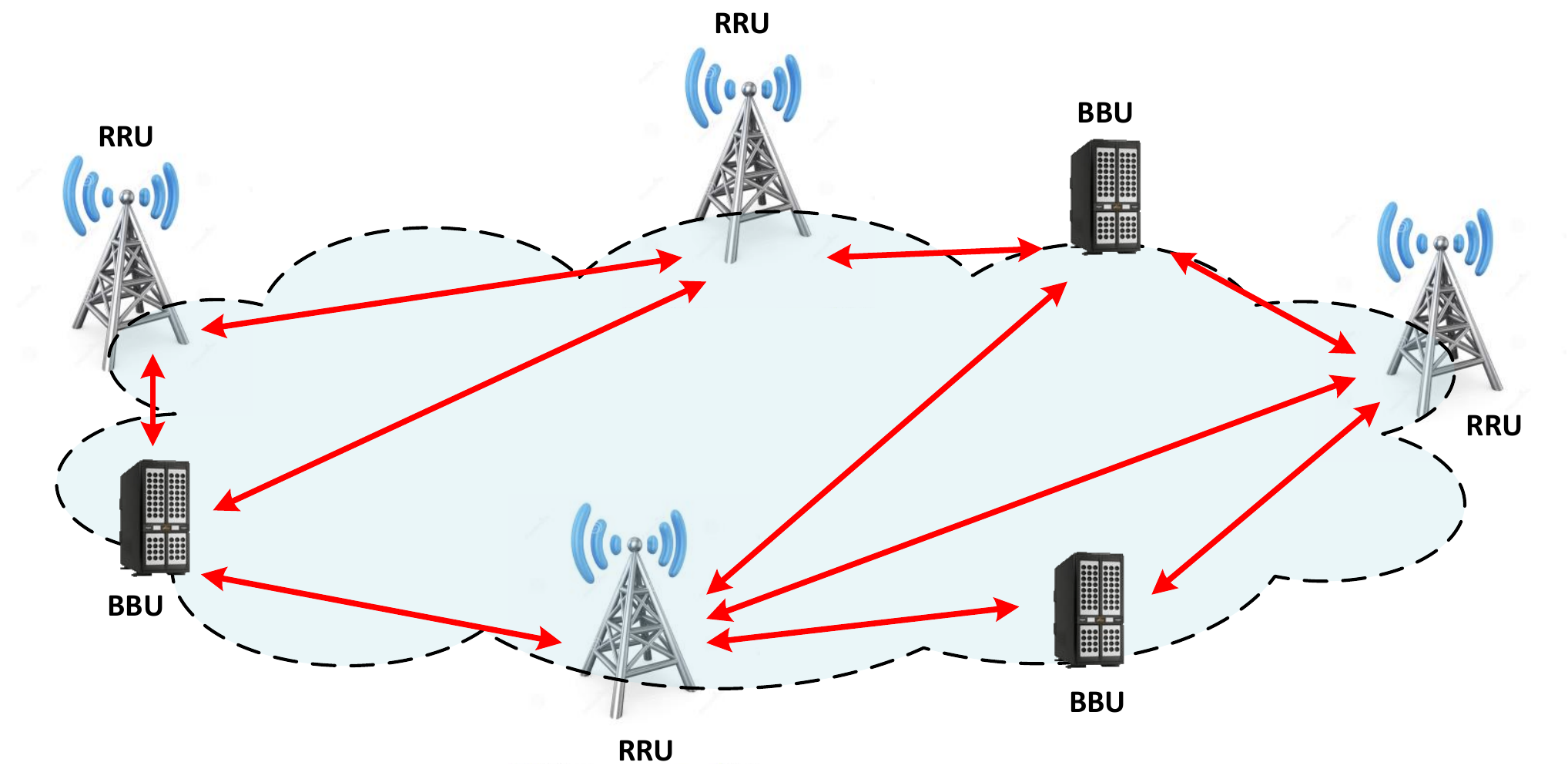}
\centering
\caption{Illustration of CPRI links (represented by arrows) within a C-RAN system.}
\label{fig:cpri}
\end{figure}

CPRI link compression techniques commonly found in the literature are based on scalar quantization \cite{samardzija2012compressed, ren2014compression, nieman2013timedomain, park2013robust}, often with block scaling before quantization to adjust for the dynamic range of samples to be quantized. For LTE networks, time-domain LTE OFDM I/Q samples are carried over CPRI links. Cyclic prefix removal for downlink and decimation to remove the inherent redundancy in oversampled LTE signals have been proposed to achieve additional compression. \cite{samardzija2012compressed} proposed a non-linear scalar quantization, whereby the quantizer is trained off-line using an iterative gradient algorithm. Compression gain of $3$ times was reported with approximately $2\%$ EVM for $10$ MHz downlink LTE data. \cite{ren2014compression} reported a compression gain of $3.3$ times with approximately $2\%$ EVM distortion using decimation, an enhanced block scaling and a uniform quantizer. Lloyd-Max scalar quantization with noise shaping was considered in \cite{nieman2013timedomain}, and distributed compression was investigated in \cite{park2013robust}.

Our approach differs from previous approaches in that we consider a vector quantization (VQ) based compression, rather than a scalar one. We exploit the fact that due to the IFFT (FFT) operation for downlink (uplink), the I/Q samples of an OFDM symbol are correlated over time. Scalar quantizer is not capable of exploiting such time correlations. On the other hand, vector quantization, by mapping grouped samples into codewords, can explore such correlations and achieves better compression gain \cite{gersho1992vector}. From complex I/Q samples, vectors need to be formed before quantization. There are several vectorization methods depending on how I/Q samples are placed within the vectors that are formed. We investigate the performance of different vectorization methods for constructing vectors from I/Q samples. For vector quantization codebook training, Lloyd algorithm is introduced. To further enhance performance, we propose a modified algorithm with different initialization step, where multiple trials work in serial to generate better codebook. Low-complexity vector quantization algorithms in the form of multi-stage vector quantization (MSVQ) and unequally protected multi-group quantization (UPMGQ) are also considered in the paper. Analysis and simulation result show that the proposed compression scheme can achieve $4$ times compression for uplink and $4.5$ times compression for downlink within $2\%$ EVM distortion.

The rest of this paper is organized as follows. Section~\ref{sec:algorithm} describes the CPRI compression algorithm in details, including the report on the serial initialization pattern to improve the performance of trained codebook from regular vector quantization as well as the discussion on universal compression. Section~\ref{sec:advanced} discusses advanced quantization methods, in order to reduce the searching and storing complexities of regular vector quantization. After that, Section~\ref{sec:simulation} contains all simulation results for both downlink and uplink, comparing different quantization methods mentioned in previous sections. Finally, Section~\ref{sec:conclusion} concludes the paper.


\section{CPRI Compression Algorithm}
\label{sec:algorithm}

A system framework for our vector quantization based CPRI compression and decompression for both downlink and uplink is illustrated in Fig.~\ref{fig:framework}. For downlink, the input to the CPRI compression module located at the BBU site is a stream of digital I/Q samples from the BBU. The CPRI compression module further contains modules of \emph{Cyclic Prefix Removal}, \emph{Decimation}, \emph{Block Scaling}, \emph{Vector Quantizer} and \emph{Entropy Encoding}. At the RRU site, the CPRI decompression module performs the reverse operations. For uplink, the analog-to-digital converter (ADC) output is the input to the CPRI compression module located at the RRU site and the CPRI decompression module at the BBU site performs the reverse operations.

\begin{figure*}[t!]
\includegraphics[scale=0.6]{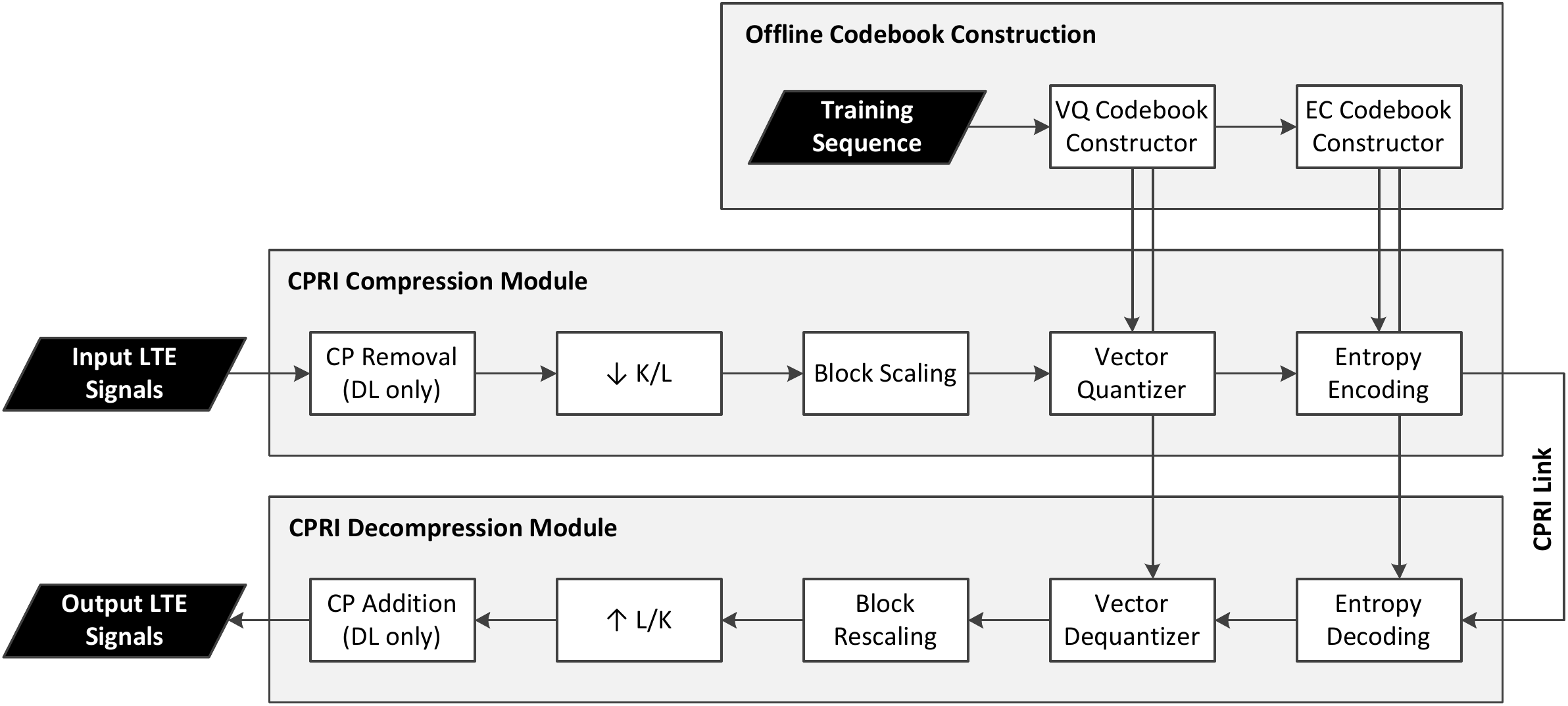}
\centering
\caption{Vector quantization based CPRI compression algorithm framework.}
\label{fig:framework}
\end{figure*}

Within these function blocks, \emph{Cyclic Prefix Removal}, \emph{Decimation}, and \emph{Block Scaling} are standard signal processing \cite{oppenheim1989discrete}. A sketch of their roles and compression gains are summarized as follows (please refer to Appendix~\ref{app:details} for details of theses blocks):
\begin{itemize}
\item \emph{CP Removal} block, applicable for downlink only, aims to eliminate the time domain redundancy from cyclic prefix. The compression gain from this block (i.e., $\textrm{CR}_{\textrm{CPR}}$) can be expressed as
\begin{align}
\textrm{CR}_{\textrm{CPR}}=\frac{L_{\textrm{SYM}}+L_{\textrm{CP}}}{L_{\textrm{SYM}}},\label{fun:CR_CPR}
\end{align}
where $L_{\textrm{SYM}}$ and $L_{\textrm{CP}}$ denote IFFT output symbol length and cyclic prefix length, respectively.
\item \emph{Decimation} block aims to reduce the redundancy in frequency domain because LTE signal is oversampled. The compression gain from this block (i.e., $\textrm{CR}_{\textrm{DEC}}$) can be expressed as
\begin{align}
\textrm{CR}_{\textrm{DEC}}=\frac{L}{K},\label{fun:CR_DEC}
\end{align}
where $L$ and $K$ denote downsampling and upsampling factors, respectively.
\item \emph{Block Scaling} block aims to lower the resolution of the signal and maintain the dynamic range to be consistent with the downstream quantization codebook. There is no direct compression gain from this block. In contrast, extra signaling overhead of $Q_{\textrm{BS}}$ bits for every $N_{\textrm{BS}}$ samples is required, where $Q_{\textrm{BS}}$ is the target resolution and $N_{\textrm{BS}}$ is the number of samples forming a block.
\end{itemize}

\emph{Vector Quantization} and \emph{Entropy Coding} blocks are key techniques in our CPRI compression algorithm. In Section~\ref{sec:algorithm:vq} and Section~\ref{sec:algorithm:ec}, we present their mechanisms and performances in detail.

\subsection{Vector Quantization}
\label{sec:algorithm:vq}

As shown in Fig.~\ref{fig:quantization}, vector quantization/dequantization is performed based on a vector quantizer codebook \cite{gray1984vector}. The inputs to vector quantization module are the vectorized samples $\bm{s}_{\textrm{VEC}}(\bar{m})$ ($\bar{m}\in\{1,\ldots,2M/L_{\textrm{VQ}}\}$, where $M$ is the number of I/Q samples and $L_{\textrm{VQ}}$ is the vector length), and the vector quantizer codebook which is a set of vector codewords $\bm{c}(k)$ ($k\in\{1,\ldots,2^{L_{\textrm{VQ}}\cdot Q_{\textrm{VQ}}}\}$, where $2^{L_{\textrm{VQ}}\cdot Q_{\textrm{VQ}}}$ is the codebook size). The codebook is trained off-line using training samples such that a specified distortion metric (such as the Euclidean distance) is minimized. The vector quantizer  maps a vector sample $\bm{s}_{\textrm{VEC}}(\bar{m})$ to one of the vector codewords which would minimize the specified distortion metric. Each quantized sample vector is represented by $L_{\textrm{VQ}}\cdot Q_{\textrm{VQ}}$ bits. Then, the compression gain from vector quantization is given by
\begin{align}
\textrm{CR}_{\textrm{VQ}}=\frac{Q_{0}}{Q_{\textrm{VQ}}},\label{fun:CR_VQ}
\end{align}
where $Q_{0}$ is the uncompressed bitwidth of I or Q component of each sample, and $Q_{\textrm{VQ}}$ is the effective bitwidth of quantized samples. $Q_{0}$ is typically $15$, specified in \cite{CPRI2011}. Next, we discuss in detail the vectorization step and the codebook training.

\begin{figure}[t!]
\includegraphics[scale=0.6]{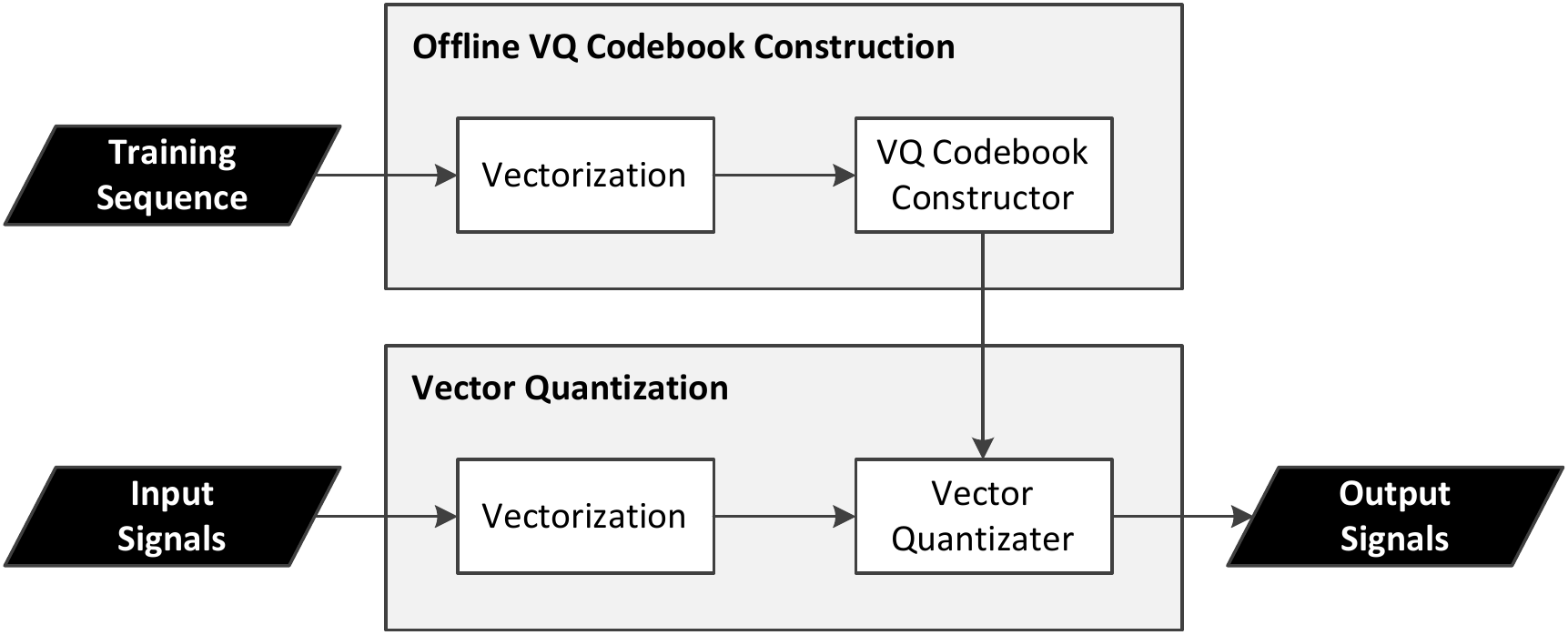}
\centering
\caption{Process chain in the \emph{Vector Quantization} block for CPRI compression.}
\label{fig:quantization}
\end{figure}

\subsubsection{Vectorization}
Before vector quantization, the vectorization of I/Q samples is needed. The purpose of vectorization is to construct vectors from I/Q samples that can capture correlation or dependency across I/Q samples. A question of interest is: given a vector length $L_{\textrm{VQ}}$, what is the best way to perform vectorization of I/Q samples such that the compression gain of vector quantization can be maximized? To this end, we consider the following three vectorization methods (see Fig.~\ref{fig:vectorization} for an illustration):
\begin{itemize}
  \item Method 1: Consecutive I components in time are grouped as vectors. Similarly, consecutive Q components are grouped as vectors.
  \item Method 2: I and Q components of the same time index are grouped as vectors.
  \item Method 3: I and Q components of all samples are randomly grouped as vectors.
\end{itemize}

\begin{figure}[t!]
\hspace*{3cm}
\includegraphics[scale=0.6]{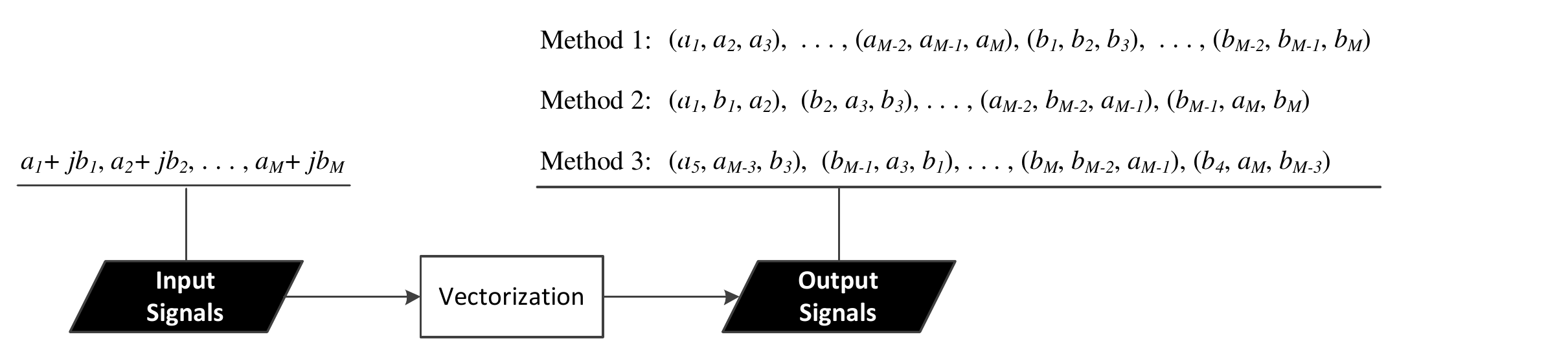}
\centering
\caption{Illustration of three vectorization methods.}
\label{fig:vectorization}
\end{figure}

The three vectorization methods are compared using the entropy of the distribution of the constructed vectors of chosen length in Euclidean orthants as metric. Since lower entropy value indicates higher correlation between the components of the constructed vectors, the vectorization method with smaller entropy implies higher quantization compression gain. Fig.~\ref{fig:vecmethod} shows plots of entropy versus $L_{\textrm{VQ}}$ for the three vectorization methods. It is evident from the plots that for LTE signals, Method 1 has the smallest entropy among the three methods, which shall be assumed in the rest of this paper.

\begin{figure}[t!]
\includegraphics[scale=0.6]{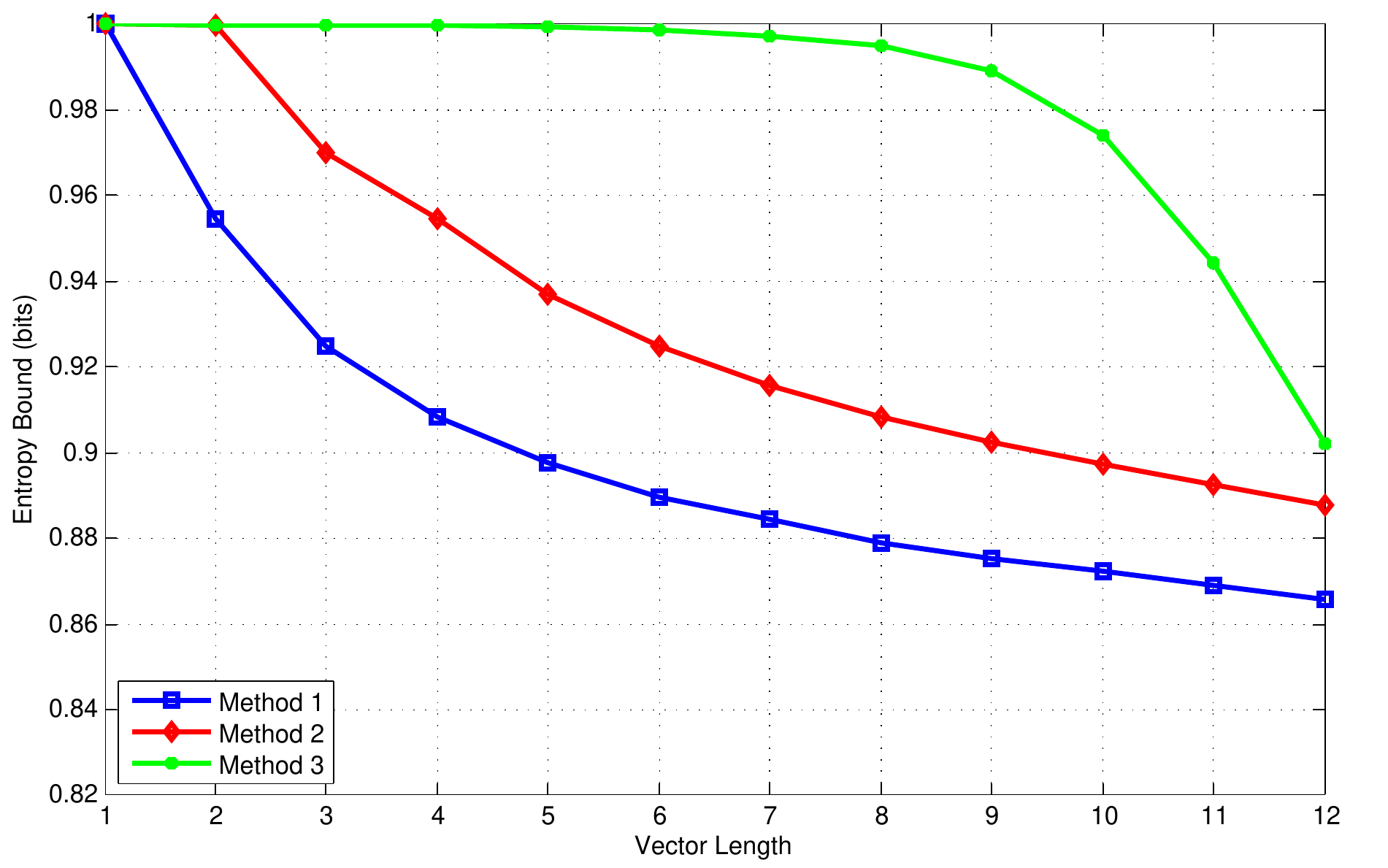}
\centering
\caption{Entropy bound verses vector length.}
\label{fig:vecmethod}
\end{figure}

\subsubsection{Codebook Training}

Our vector quantizer codebook is trained using Lloyd algorithm \cite{lloyd1982least} (a special case of LBG algorithm \cite{linde1980algorithm}). Lloyd algorithm is an iterative algorithm that can be utilized to construct codebook for vector quantization. It aims to find evenly-spaced sets of points (as codewords) in subsets of Euclidean spaces, and to partition input samples into well-shaped and uniformly sized convex cells. In general, Lloyd algorithm may start by randomly picking a number of input samples as initial codewords, and it then repeatedly executes samples partitioning and codebook updating in every iteration to reduce target distortion. A commonly used distortion metric is the Euclidean distance. Each time, the codeword points are left in a slightly more even distribution: closely spaced points move farther apart, and widely spaced points move closer together. Finally, Lloyd algorithm terminates at certain local optimal, after a proper stopping criterion is satisfied.

Classical Lloyd algorithm is known to be quite sensitive to the initial choice of codewords, especially when input sample space does not have a smooth structure \cite{lloyd1982least}\cite{bishop2006pattern}. In other words, if the input samples are concentrated in a particular space, Lloyd algorithm easily converges to a local optimum which can be away from the global optimum. To this end, classical Lloyd algorithm performs multiple independent trials and chooses the codebook that produces the lowest distortion metric from all trials to evade the initialization problem.
A process chain for classical Lloyd algorithm with multiple trials is illustrated in Fig.~\ref{fig:classical_lloyd}.

\begin{figure}[t!]
\includegraphics[scale=0.6]{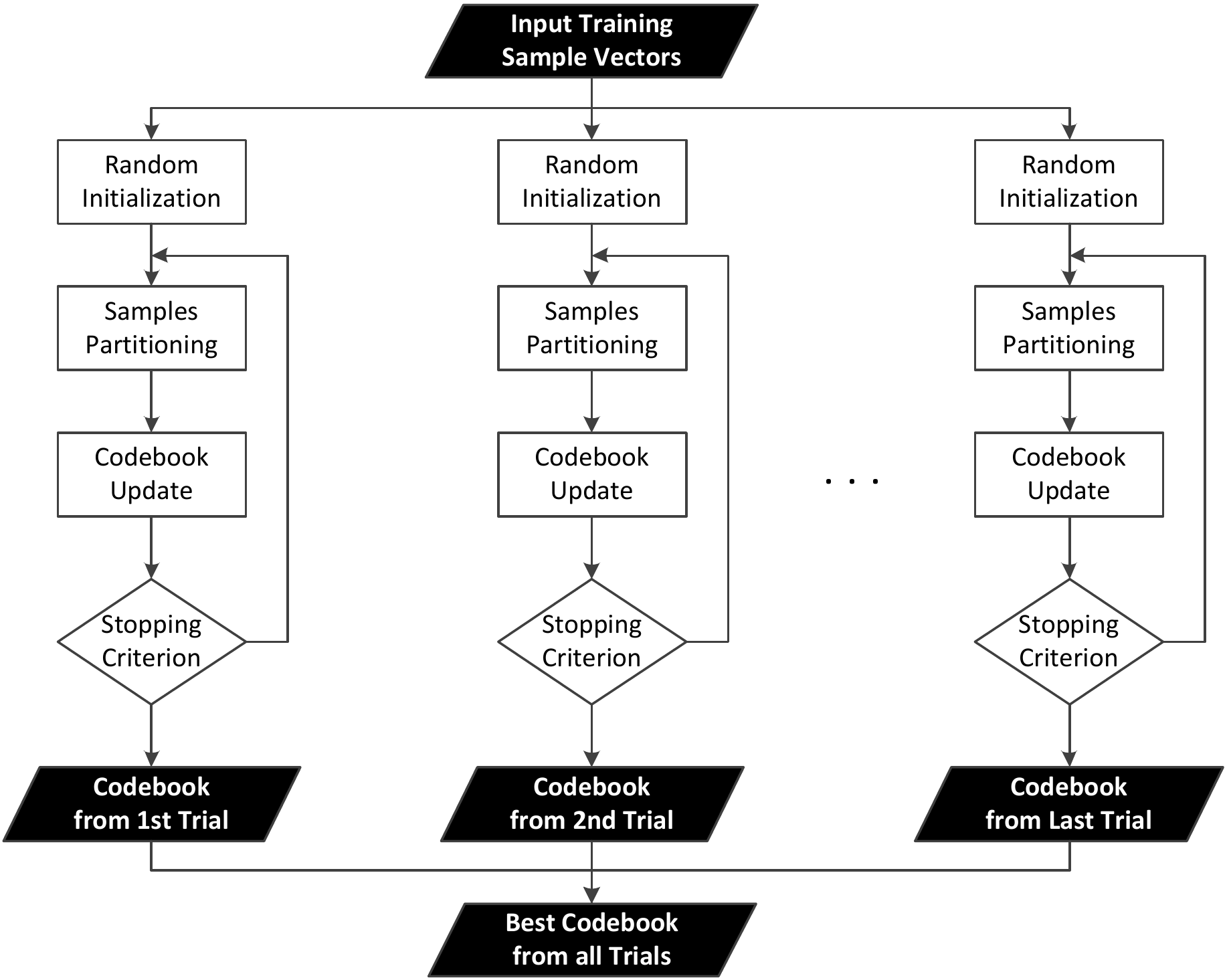}
\centering
\caption{Process chain for classical Lloyd algorithm with multiple trials working in parallel to perform vector quantization.}
\label{fig:classical_lloyd}
\end{figure}

In the case of LTE I/Q samples, the sample values are observed to be highly concentrated in a narrow range. However, we also observe that independent trials is not effective to overcome the initialization problem described earlier, specifically we observe persistent convergence to similar local optimum despite multiple Lloyd trials. To resolve this issue, we introduce a modified Lloyd algorithm, whereby the output (codebook) from previous trial is utilized as the input (the initial codebook) to the next trial, after applying proper rescaling to the initial codebook magnitude. The procedure of rescaling is essential, because the output from previous trial is already a local optimum. Rescaling helps to evade this local optimum and restart the search for a better codebook. The rescaling factor can be the square root of the average power of the I/Q samples. The diagram illustrating the modified Lloyd algorithm is given in Fig.~\ref{fig:modified_lloyd}.

\begin{figure}[t!]
\includegraphics[scale=0.6]{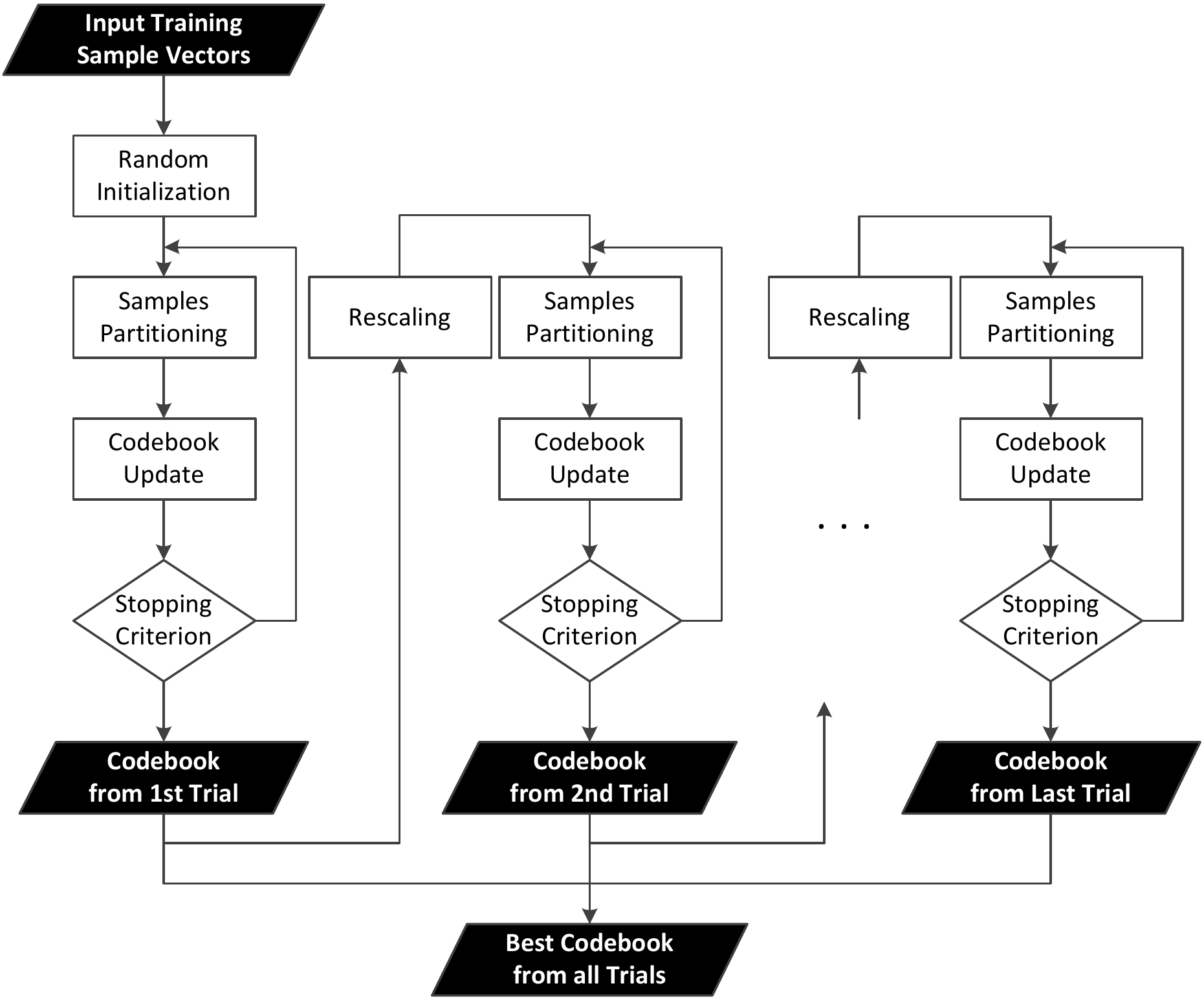}
\centering
\caption{Process chain for modified Lloyd algorithm, where multiple trials work in serial.}
\label{fig:modified_lloyd}
\end{figure}

The effectiveness of the modified Lloyd algorithm is illustrated by comparing the codebooks trained using the classical Lloyd algorithm and the modified Lloyd algorithm as shown in Fig.~\ref{fig:codebook_compare_classical} and Fig.~\ref{fig:codebook_compare_modified} respectively. In this simulation, LTE uplink signals with $5$dB SNR, $64$QAM modulation, and AWGN channel are considered as training sequences using VQ with $Q_{\textrm{VQ}}=6$ and $L_{\textrm{VQ}}=2$. Two axes of Fig.~\ref{fig:codebook_compare} represent the two elements of the grouped vectors. Each point on the plot represents a codeword and the color corresponds to frequencies (warmer color for higher frequency). It is observed that the codebook from modified algorithm better reflects the distribution of training sample vectors. This is also confirmed from the EVM improvement of $3.10\%$ from the classical algorithm to $2.54\%$ from the modified Lloyd algorithm for this particular case ($18.1\%$ improvement).

\begin{figure*}[t!]
    \centering
    \begin{subfigure}[b]{0.30\textwidth}
        \centering
        \includegraphics[width=0.985\textwidth]{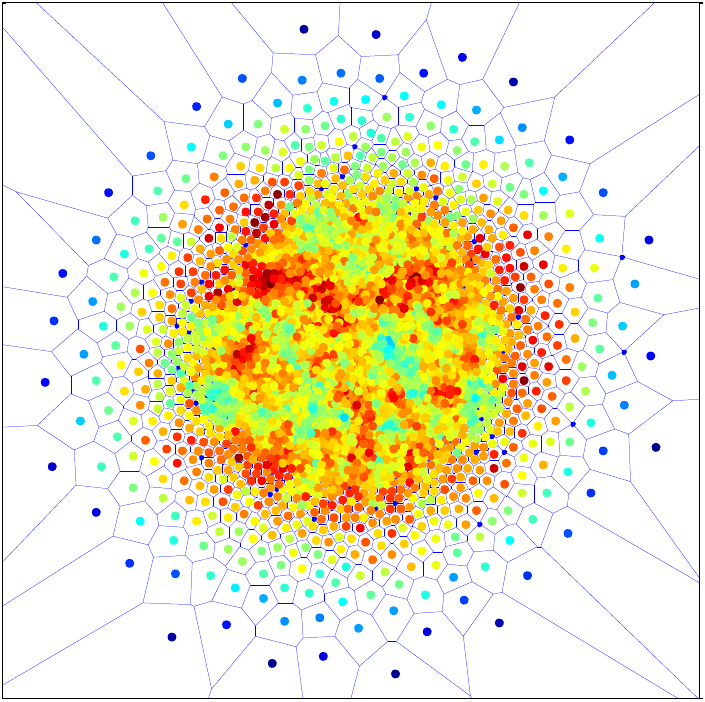}
        \caption{VQ codebook with Classical Lloyd algorithm.}
        \label{fig:codebook_compare_classical}
    \end{subfigure}
    \hfill
    \begin{subfigure}[b]{0.30\textwidth}
        \centering
        \includegraphics[width=1\textwidth]{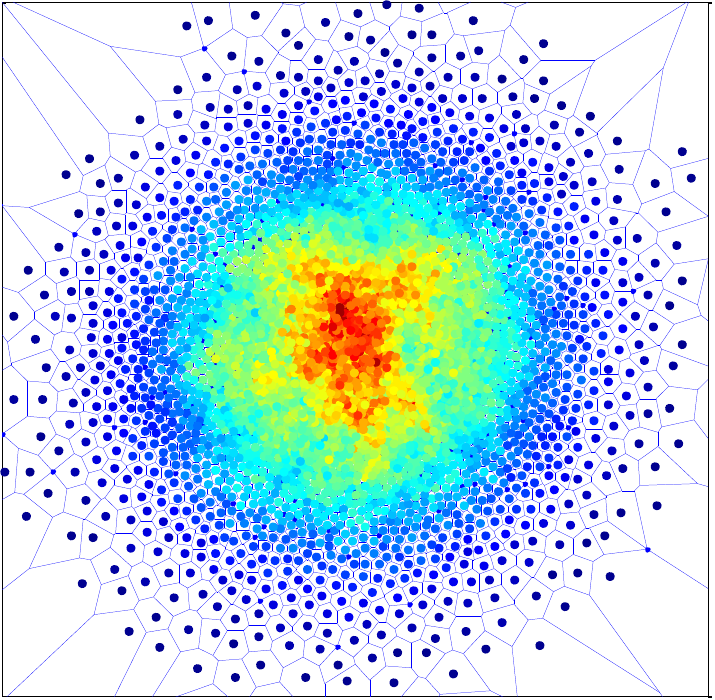}
        \caption{VQ codebook with Modified Lloyd algorithm.}
        \label{fig:codebook_compare_modified}
    \end{subfigure}
    \hfill
    \begin{subfigure}[b]{0.30\textwidth}
        \centering
        \includegraphics[width=0.985\textwidth]{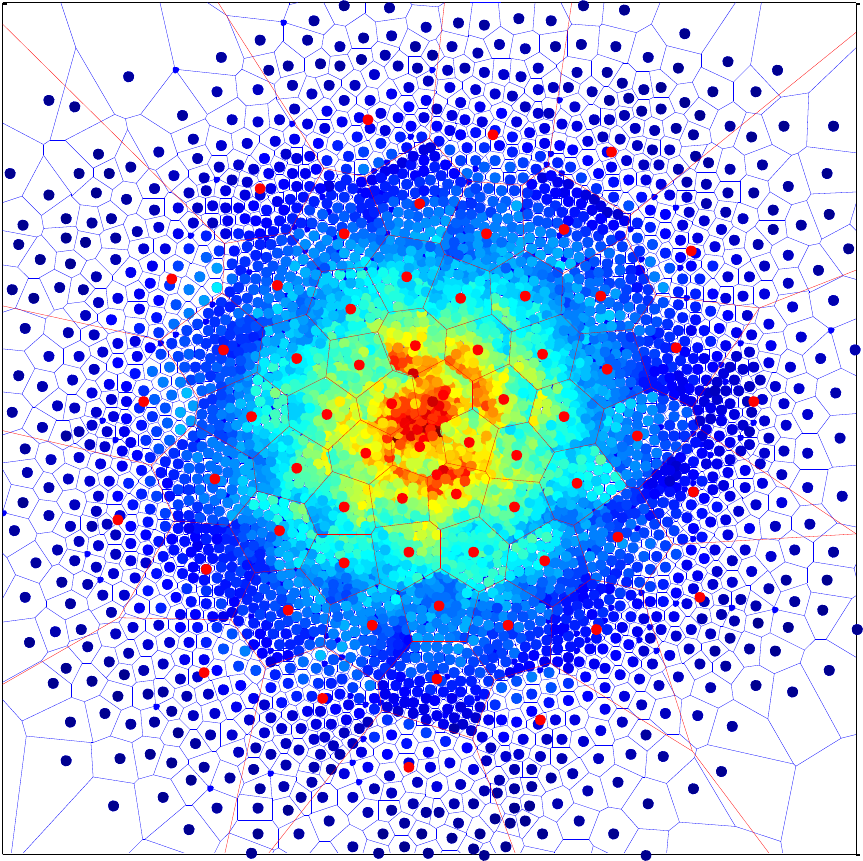}
        \caption{MSVQ codebook with Modified Lloyd algorithm.}
        \label{fig:codebook_compare_MSVQ}
    \end{subfigure}
    \caption{Comparison of codebooks from classical Lloyd algorithm, modified Lloyd algorithm, and MSVQ respectively.}
    \label{fig:codebook_compare}
\end{figure*}

\subsection{Entropy Coding}
\label{sec:algorithm:ec}

Entropy coding is a lossless data compression scheme that utilizes more bits to represent sources with lower frequency (higher entropy) and fewer bits to represent sources with higher frequency (lower entropy). The most common entropy coding technique is Huffman coding \cite{huffman1952method}. In Huffman coding, the dictionary construction procedure can be completed in linear time, and the final codeword for each symbol can be simply constructed from a splitting tree. According to Shannon's lossless source coding theorem \cite{shannon1948mathematical}, the optimal average coded length for a source is its entropy, and Huffman codes are proven to be optimal with linear time complexity for symbol-by-symbol coding \cite{cover2006elements}.

Without entropy coding, each codeword of vector quantization codebook requires exactly $L_{\textrm{VQ}}\cdot Q_{\textrm{VQ}}$ bits for representation. However, it is observed that the probability mass function (PMF) of codewords is not uniform (see Fig.~\ref{fig:codebook_compare}, where warmer color represents higher probability), implying potential entropy coding gain. To this end, Huffman coding is applied based on the PMF of codewords. Denoting the average length of Huffman codes as $L_{\textrm{HUFF}}$, the compression gain from entropy coding can be expressed as
\begin{align}
\textrm{CR}_{\textrm{EC}}=\frac{L_{\textrm{VQ}}\cdot Q_{\textrm{VQ}}}{L_{\textrm{HUFF}}}.\label{fun:CR_EC}
\end{align}

\subsection{Performance Evaluation}

In this subsection, the aforementioned blocks in system are integrated together and performance of the whole CPRI compression framework is evaluated. The compression gain and the EVM distortion are investigated. We summarize the main results of the proposed CPRI compression algorithm in a theorem as follows.
\begin{theorem}\label{thm:CR}
For a given vectorization method and a given vector length, the proposed vector quantization based lossy compression algorithm for CPRI link can achieve the rate distortion trade-off given by
\emph{
\begin{align}
\textrm{CR}_{\textrm{CPRI}}=\frac{1}{1/(\textrm{CR}_{\textrm{CPR}}\cdot \textrm{CR}_{\textrm{DEC}}\cdot \textrm{CR}_{\textrm{VQ}}\cdot \textrm{CR}_{\textrm{EC}})+Q_{\textrm{BS}}/(2Q_{0} \cdot N_{\textrm{BS}})},\label{fun:CR}
\end{align}}
and
\emph{\begin{align}
\textrm{EVM}_{\textrm{TD}}(\%)\triangleq\sqrt{\frac{\sum_{m=1}^{M}|\bm{s}_{\textrm{IN}}(m)-\bm{s}_{\textrm{OUT}}(m)|^2}{\sum_{m=1}^{M}|\bm{s}_{\textrm{IN}}(m)|^2}}\times 100,\label{fun:TD_EVM}
\end{align}}
for large value of $M$, where the compression gains in the denominator of \eqref{fun:CR} are given by \eqref{fun:CR_CPR}, \eqref{fun:CR_DEC}, \eqref{fun:CR_VQ}, and \eqref{fun:CR_EC}, respectively;  $Q_{0}$ is the uncompressed bitwidth of I or Q component of each sample; $Q_{\textrm{BS}}$ and $N_{\textrm{BS}}$ are parameters for block scaling; $\bm{s}_{\textrm{IN}}$ is the input sequence to compression algorithm; $\bm{s}_{\textrm{OUT}}$ is the output from decompression algorithm; $M$ is the number of input or output complex samples.
\end{theorem}
\begin{proof}
From the description of proposed algorithm, if $M$ complex samples are considered as input to the module, where each complex component is represented by $Q_{0}$ bits, after compression, $2M/(\textrm{CR}_{\textrm{CPR}}\cdot \textrm{CR}_{\textrm{DEC}}\cdot L_{\textrm{VQ}})$ number of binary strings are transmitted on CPRI link, where the average length for each string is $L_{\textrm{HUFF}}$. For \emph{Block Scaling} block, $Q_{\textrm{BS}}$ number of extra bits are needed for every  $N_{\textrm{BS}}$ complex samples. Hence, the final compression gain by CPRI compression can be expressed as
\begin{align}
\textrm{CR}_{\textrm{CPRI}} &\triangleq \frac{\textrm{number of bits input to the module}}{\textrm{number of bits transmitted on CPRI link}}\nonumber\\
&=\frac{2Q_{0} \cdot M}{L_{\textrm{HUFF}}\cdot 2M/(\textrm{CR}_{\textrm{CPR}}\cdot \textrm{CR}_{\textrm{DEC}}\cdot L_{\textrm{VQ}})+Q_{\textrm{BS}}\cdot M/N_{\textrm{BS}}},\nonumber
\end{align}
which further gives the expression in \eqref{fun:CR}.  Note that if any of the blocks is not enabled/present, the corresponding compression gain is set to $1$ (especially, if the block scaling is not enabled, $Q_{\textrm{BS}}$ is set to $0$).

Equation \eqref{fun:TD_EVM} is the standard form of error vector magnitude (EVM), which is a measure of deviation of constellation points from their ideal locations. In this study, EVM is used to quantify the distortions introduced by compression.
\end{proof}

Time-domain EVM calculates the distortion over the whole bandwidth. However, for LTE signals, only part of the bandwidth carries useful information. This motivates the use of frequency-domain EVM instead of the time-domain EVM \eqref{fun:TD_EVM} as the distortion measure, i.e.,
\begin{align}
\textrm{EVM}_{\textrm{FD}}(\%)\triangleq\sqrt{\frac{\sum_{n\in\mathcal{B}}|\tilde{\bm{s}}_{\textrm{IN}}(n)-\tilde{\bm{s}}_{\textrm{OUT}}(n)|^2}{\sum_{n\in\mathcal{B}}|\tilde{\bm{s}}_{\textrm{IN}}(n)|^2}}\times 100,\label{fun:FD_EVM}
\end{align}
where $\tilde{\bm{s}}_{\textrm{IN}}(n)$ and $\tilde{\bm{s}}_{\textrm{OUT}}(n)$ are transformed signals by FFT for input and output respectively, and $\mathcal{B}$ is the collection of indices corresponding to the utilized bandwidth.

\subsection{Universal Compression}

In principle, vector quantization codebook constructed from a particular set of training samples can produce the intended performance if the real samples do not deviate significantly from the training samples in statistical sense. However, in practice, system parameters or properties, such as channel types, SNRs, and modulation schemes, may not be known perfectly, or may not always match those assumed for the training samples. This motivates the need for a universal compression technique, i.e., a robust codebook that can provide reasonable or acceptable performances in all or a large range of system parameters, although the codebook is not necessarily optimal for a specific system setting. To understand if such a universal codebook exists, we investigate the performance of codebooks produced using mismatched training samples, in order to find the sensitivity factors impacting the performance. A key step is to discover the most robust parameter contributing to the universal codebook. If this procedure is infeasible, i.e., such a parameter does not exist, a larger training sample pool mixing with different parameters should be considered, such that the trained codebook could reflect the distributions of all parameters.

For example, for uplink LTE samples with $4$ times target compression ratio, if vector length $L_{\textrm{VQ}}=2$ with decimation value $5/8$ is performed, constructed codebook is rather insensitive to modulation methods or channel types (see TABLE~\ref{tab:level2_mismatch_modulation} and TABLE~\ref{tab:level2_mismatch_channel}), but is relatively more sensitive to SNRs (see TABLE~\ref{tab:level2_mismatch_SNR}). Nevertheless, the $5$dB codebook can be adopted as the universal codebook, since its EVM performances are acceptable for all SNR training samples. Based on these observations, we conclude that the codebook obtained from $5$dB SNR, $64$QAM modulation, and AWGN channel can be a suitable universal codebook, which we shall assume for further performance evaluation in Section~\ref{sec:simulation:ul}.

\begin{table}[t!]
\caption{SNRs mismatch for uplink samples implemented with vector quantization ($L_{\textrm{VQ}}=2$). Relative performance marked in percentage is compared per column.}
\centering
\begin{tabular}{c | c c c c c}
\toprule
EVM $(\%)$	&\multicolumn{5}{c}{Evaluated samples}	\\
\midrule
\multirow{4}{*}{Training samples}	   & 	    &$0$dB	&$5$dB & $10$dB	&$20$dB	\\
	                                   &$0$dB	&{\color{red} $2.55~(\,\,\,-\,\,\,)$}	&$2.61~(\,\,3\%\,)$	&$2.72~(\,\,7\%\,)$	& $2.89~(15\%)$\\
	                                   &$5$dB	&$2.68~(\,\,5\%\,)$	&{\color{red} $2.54~(\,\,\,-\,\,\,)$}	&$2.56~(\,\,1\%\,)$	& $2.65~(\,\,5\%\,)$\\
                                       &$10$dB  &$3.11~(22\%)$ &$2.89~(14\%)$   &{\color{red} $2.54~(\,\,\,-\,\,\,)$}   &$2.56~(\,\,2\%\,)$ \\
	                                   &$20$dB	&$3.41~(34\%)$	&$2.98~(17\%)$	&$2.69~(\,\,6\%\,)$	&{\color{red} $2.52~(\,\,\,-\,\,\,)$}\\
\bottomrule
\end{tabular}
\label{tab:level2_mismatch_SNR}
\end{table}

\begin{table}[t!]
\caption{Modulations mismatch for uplink samples implemented with vector quantization ($L_{\textrm{VQ}}=2$). Relative performance marked in percentage is compared per column.}
\centering
\begin{tabular}{c | c c c c }
\toprule
EVM $(\%)$	&\multicolumn{4}{c}{Evaluated samples}	\\
\midrule
\multirow{4}{*}{Training samples}	&	&QPSK	&$16$QAM	&$64$QAM	\\
	&QPSK	&{\color{red} $2.54~(\,\,-\,\,)$}	&$2.57~(1\%)$	&$2.64~(4\%)$	\\
	&$16$QAM	&$2.54~(0\%)$	&{\color{red} $2.54~(\,\,-\,\,)$}	&$2.60~(2\%)$	\\
	&$64$QAM	&$2.56~(1\%)$&$2.56~(1\%)$	&{\color{red} $2.54~(\,\,-\,\,)$}	\\
\bottomrule
\end{tabular}
\label{tab:level2_mismatch_modulation}
\end{table}

\begin{table}[t!]
\caption{Channel types mismatch for uplink samples implemented with vector quantization ($L_{\textrm{VQ}}=2$). Relative performance marked in percentage is compared per column.}
\centering
\begin{tabular}{c | c c c }
\toprule
EVM $(\%)$	&\multicolumn{3}{c}{Evaluated samples}	\\
\midrule
\multirow{3}{*}{Training samples}	&	&AWGN	&Ped B	\\
	&AWGN	&{\color{red} $2.54~(\,\,-\,\,)$}	&$2.56~(2\%)$	\\
	&Ped B	&$2.64~(4\%)$	&{\color{red} $2.52~(\,\,-\,\,)$}		\\
\bottomrule
\end{tabular}
\label{tab:level2_mismatch_channel}
\end{table}

However, if we consider the same system setup but assuming $L_{\textrm{VQ}}=3$, the constructed codebook is quite sensitive to SNRs, such that no codebook based on a single SNR could be adopted as the universal codebook (see TABLE~\ref{tab:level3_mismatch_SNR}). If trained codebooks are utilized to compress mismatched target SNR samples, the distortions are even larger than the ones from $L_{\textrm{VQ}}=2$ (compare with TABLE~\ref{tab:level2_mismatch_SNR}). To solve this problem, we construct a larger training sample set, which contains subframes with diverse SNRs, and perform training over this database with larger SNR region, then the resulting codebook may not be optimal for the particular SNR, but it can achieve acceptable performance ($\sim2.1\%$ EVM in last row of TABLE~\ref{tab:level3_mismatch_SNR}, which is better than the EVM from $L_{\textrm{VQ}}=2$ in TABLE~\ref{tab:level2_mismatch_SNR}) for the whole SNR region of concern.

\begin{table}[t!]
\caption{SNRs mismatch for uplink samples implemented with vector quantization ($L_{\textrm{VQ}}=3$). Relative performance marked in percentage is compared per column.}
\centering
\begin{tabular}{c | c c c c c}
\toprule
EVM $(\%)$	&\multicolumn{5}{c}{Evaluated samples}	\\
\midrule
\multirow{6}{*}{Training samples}	&	&$0$dB &$5$dB	&$10$dB	&$20$dB	\\
	&$0$dB	&{\color{red} $1.71~(\,\,\,-\,\,\,)$}	&$2.75~(65\%)$	&$2.72~(68\%)$	&$2.72~(69\%)$\\
    &$5$dB	&$2.79~(63\%)$ &{\color{red}$1.67~(\,\,\,-\,\,\,)$}	&$2.70~(67\%)$	&$2.68~(66\%)$\\
	&$10$dB	&$2.88~(68\%)$	& $2.77~(66\%)$	&{\color{red} $1.62~(\,\,\,-\,\,\,)$}	&$2.55~(58\%)$\\
	&$20$dB	&$2.90~(70\%)$	&$2.80~(68\%)$	&$2.60~(60\%)$	&{\color{red} $1.61~(\,\,\,-\,\,\,)$}\\
    &{\color{blue} mixed SNRs} & {\color{blue} $2.17~(27\%)$} &{\color{blue}  $2.15~(29\%)$} & {\color{blue} $2.12~(31\%)$} & {\color{blue} $2.07~(29\%)$}\\
\bottomrule
\end{tabular}
\label{tab:level3_mismatch_SNR}
\end{table}

\section{Reduced Complexity Quantization Methods}
\label{sec:advanced}

The vector codebook constructed using Lloyd algorithm performs reasonably well to quantize LTE samples. However, for lower EVM requirements, the size of codebook for vector quantization should be larger, which leads to longer training time, larger storage space, and slower encoding process. This motivates the need to design a low-complexity vector quantization algorithm.

\subsection{Multi-Stage Vector Quantization}

The structured codebook with multiple quantization stages is one way to achieve complexity reduction. In this so called Multi-Stage Vector Quantization (MSVQ) scheme \cite{gersho1992vector}\cite{juang1982multiple}, each lower stage codebook (smaller size) partitions its upper stage codebook (larger size). To quantize (see Fig.~\ref{fig:multi_stage}), we start with the lowest stage codebook and obtain the quantization vector. The resultant quantized vector is then used to identify one of the partitions of the upper stage codebook for quantization. The process continues until the uppermost stage is achieved. Fig.~\ref{fig:codebook_compare_MSVQ} illustrates an example of two-stage MSVQ, where the red grid represents lower stage codebook boundaries, and blue grid represents higher stage codebook boundaries.

\begin{figure}[t!]
\includegraphics[scale=0.6]{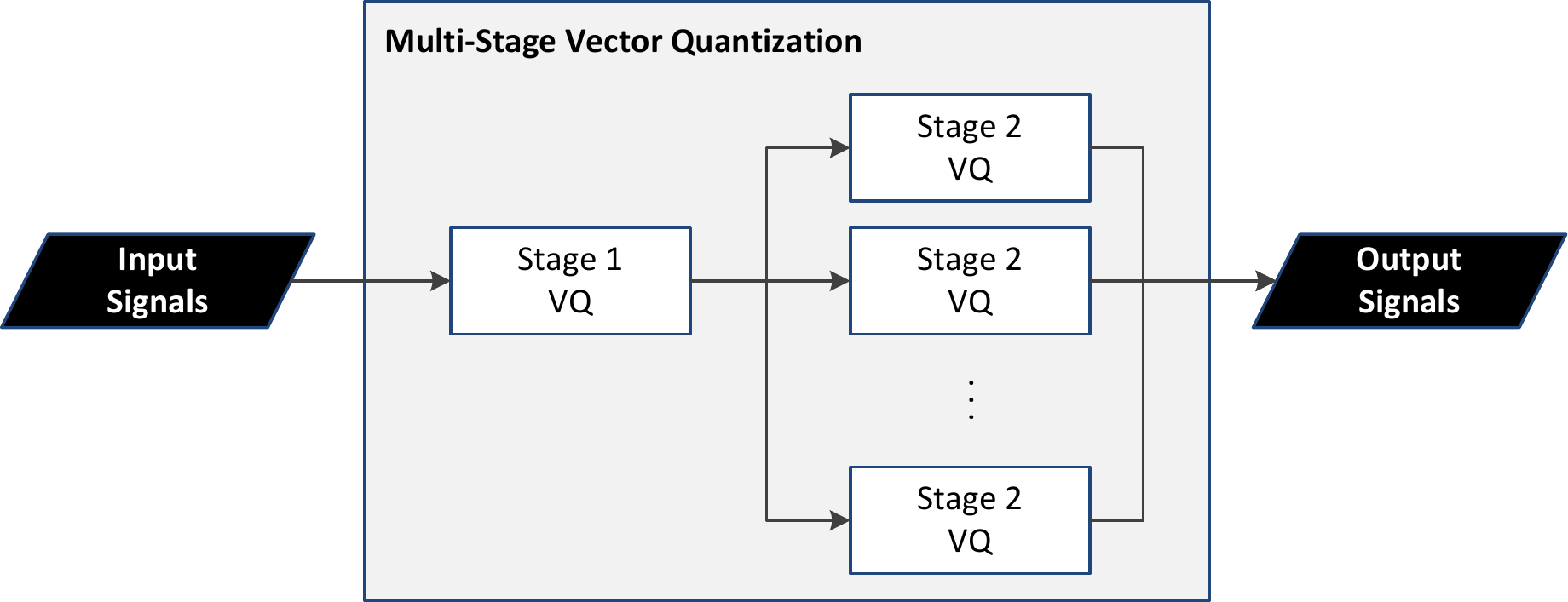}
\centering
\caption{Process chain of two-stage vector quantization. Quantization result from the first stage is utilized for second stage quantization.}
\label{fig:multi_stage}
\end{figure}

In principle, MSVQ optimizes the target distortion function in multiple stages. Hence for the same compression target, its achievable EVM cannot be better than that achieved by the classical single stage VQ in theory (see illustration of codebooks of VQ and MSVQ for the same training samples in Fig.~\ref{fig:codebook_compare}). However, we observe that the distortion performance of MSVQ can be quite close to VQ (also see simulation results in the next section). Meanwhile, the codebook search complexity and codebook training time for MSVQ is remarkably improved compared to the original VQ, due to relatively smaller codebook size in each stage. Note that MSVQ does not decrease the codebook size at the highest stage.

\subsection{Unequally Protected Multi-Group Quantization}

Although MSVQ reduces significantly the training and searching complexities of regular VQ, the uppermost stage codebook size remains the same as VQ. We propose another novel quantization method for low complexity and latency, referred as Unequally Protected Multi-Group Quantization (UPMGQ). The intuition of this quantization method comes from an expansion coding scheme for continuous-valued sources in \cite{si2014lossy}, where the target sources are expanded into independent parallel levels and unequal protection degrees are performed on expanded levels based on their importance and contributions to the distortion. Inspired by the idea of expansion, a combination of expansion coding and vector quantization scheme is proposed. More precisely, each sample can be written as a sum of powers of $2$ through binary expansion (or as a binary number), where each power of $2$ is referred to as a ``level''.
For example, the number $17.5$ can be represented as $1\times2^4 + 0\times2^3 + 0\times2^2 + 0\times2^1 + 1\times2^0 + 1\times2^{-1}$ (or $10001.1$), which takes the value $1$ for level $4$, level $0$ and level $-1$, and the value $0$ for level $3$, level $2$ and level $1$.
The sign bit and the higher levels are clearly more significant than the lower levels, with respect to low distortion compression. This observation implies that levels after expansion should not be treated equally, and more protection should given to the more significant levels.

\begin{figure}[t!]
\includegraphics[scale=0.6]{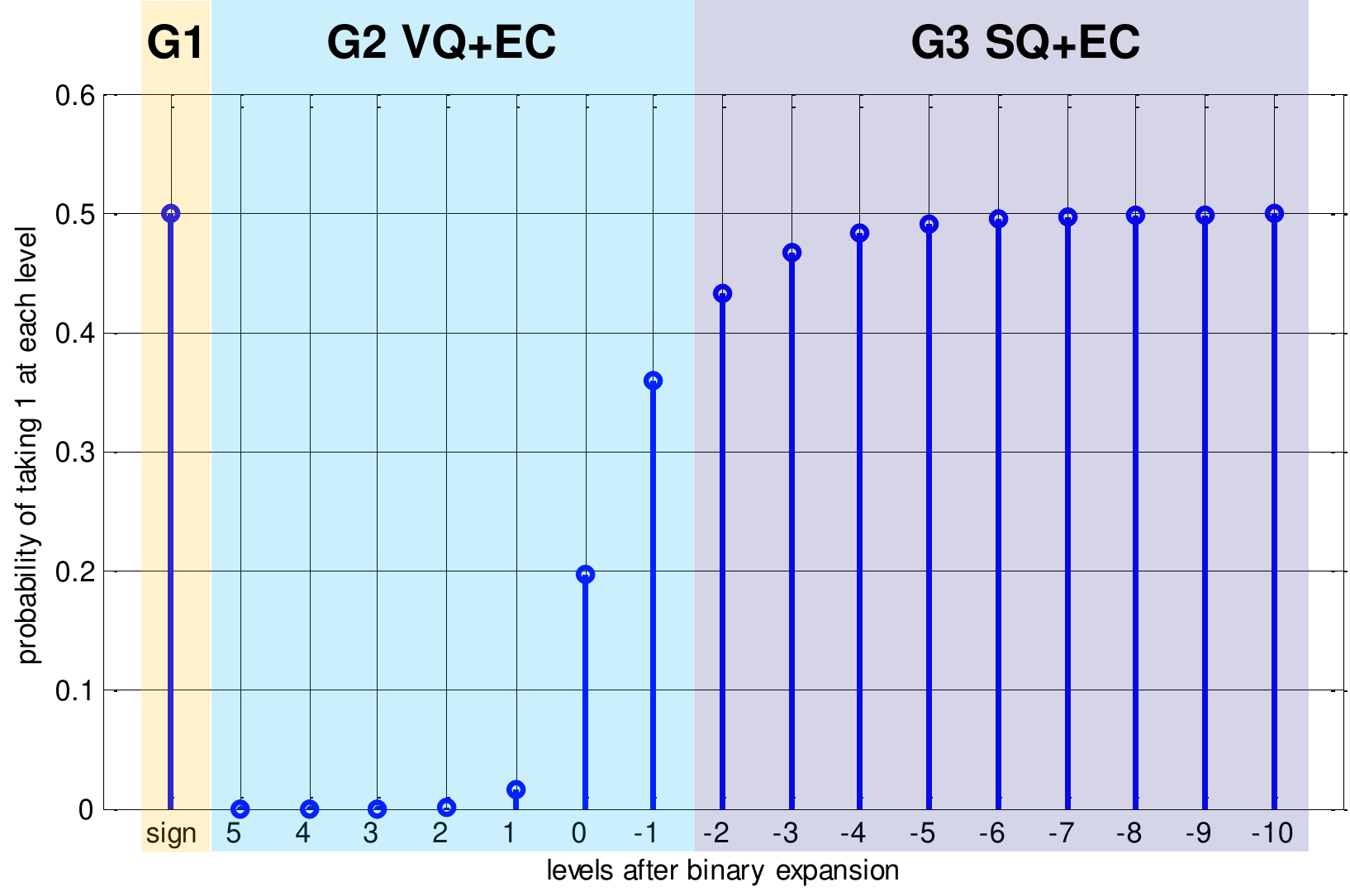}
\centering
\caption{Illustration of unequally protected multi-group quantization based on binary expansion of normalized signals.}
\label{fig:multi_group}
\end{figure}

Fig.~\ref{fig:multi_group} gives an example to illustrate the statistics of levels after binary expansion for OFDM/SC-FDM (e.g. LTE/LTE-Advanced) normalized signals. X-axis is the level index for base-$2$ expansion, and Y-axis shows the corresponding probability of taking value $1$ at each level. Evidently, the probability of taking $1$ tends to $0.5$ for the lower levels, and tends to $0$ for higher levels due to power constraint. Besides, the sign of signal samples constitutes a separate level, which approximates Bernoulli $0.5$ distribution due to the symmetry of signals. Based on this observation, in order to enable low distortion and to efficiently utilize coding rate, the levels are partitioned into multiple groups and different quantization schemes are applied to different groups:
\begin{itemize}
\item[1)]	For sign bit level ($\bm{\mathsf{G1}}$ in Fig.~\ref{fig:multi_group}), no quantization or entropy coding is performed. $1$ bit is utilized to perfectly represent the sign of signal, since any error from the sign will lead to large distortion.
\item[2)]	For higher levels above a threshold ($\bm{\mathsf{G2}}$ in Fig.~\ref{fig:multi_group}), vector quantization combining with entropy coding can be utilized to exploit the correlation among signals, and to fully compress the redundancy in codebook after vector quantization.
\item[3)]	For lower levels below the threshold ($\bm{\mathsf{G3}}$ in Fig.~\ref{fig:multi_group}), scalar quantization with entropy coding can be utilized.
For these levels, correlation among levels becomes minute due to their smaller weights in the expansion. If the threshold is chosen such that all lower levels are almost Bernoulli $0.5$ distributed, then, entropy coding may not be essential in this case.
\end{itemize}

The threshold $\theta$ that separates the higher levels and the lower levels is a design parameter, which can be tuned or optimized according to a desired objective. More precisely, for a real valued signal $s$, its higher and lower group values as defined in UPMGQ are given by
\begin{align}
&s_{\textrm{HIGH}} = 2^{\theta}\cdot\lfloor|s|\cdot2^{-\theta}\rfloor,\label{equ:high_part}\\
&s_{\textrm{LOW}} = |s|-2^{\theta}\cdot\lfloor|s|\cdot2^{-\theta}\rfloor.\label{equ:low_part}
\end{align}

The role of \emph{UPMGQ} block is to substitute \emph{Vector Quantization} and \emph{Entropy Coding} blocks in the framework of CPRI compression. A detailed description of process chain inside \emph{UPMGQ} block is shown in Fig.~\ref{fig:multi_group_framework}.
The output from UPMGQ are $3$ groups of data streams, which are further transmitted on the CPRI link. The compression gain from \emph{UPMGQ} block alone can be calculated as follows.
\begin{align}
\textrm{CR}_{\textrm{UPMGQ}}=\frac{Q_0}{1+L_{\textrm{HIGH}}/L_{\textrm{UPMGQ}}+L_{\textrm{LOW}}},\label{fun:CR_UPMGQ}
\end{align}
where $L_{\textrm{HIGH}}$ is the average length of codewords (after entropy coding, if applied) for high levels group ($\bm{\mathsf{G2}}$), and $L_{\textrm{LOW}}$ is the corresponding one for low levels group ($\bm{\mathsf{G3}}$) (scalar quantization assumed); $L_{\textrm{UPMGQ}}$ is the vector length of VQ in high levels group; the number $1$ in numerator is the rate for sign bit ($\bm{\mathsf{G1}}$) (no compression assumed).

\begin{figure}[t!]
\includegraphics[scale=0.6]{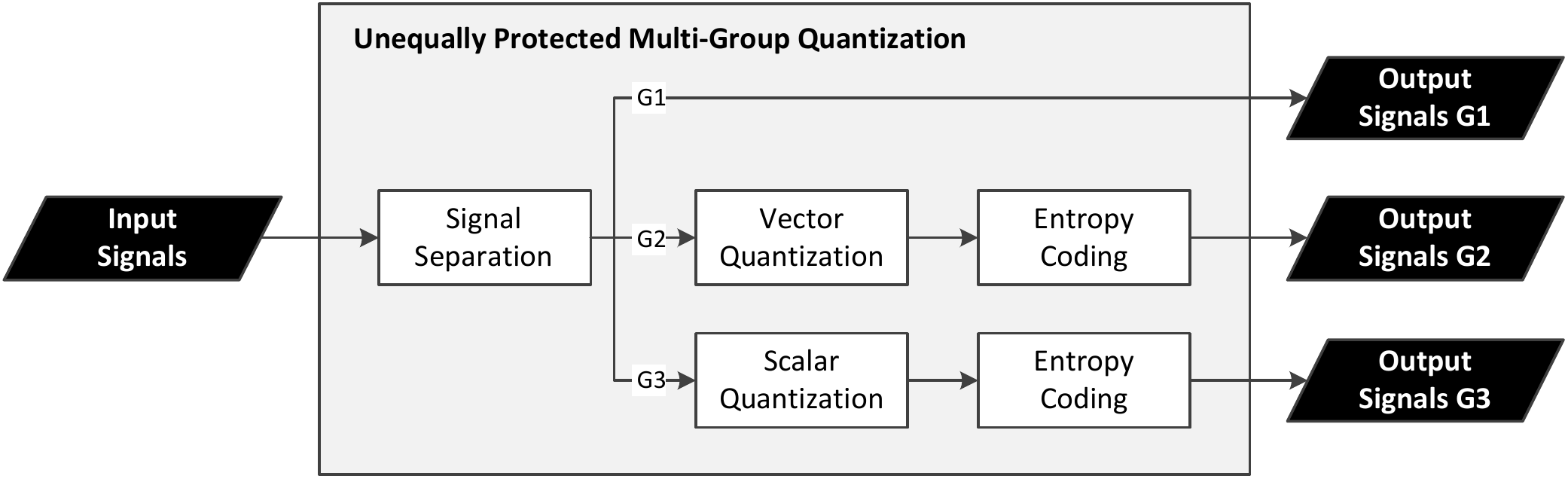}
\centering
\caption{Process chain in the multi-group quantization block for CPRI compression.}
\label{fig:multi_group_framework}
\end{figure}

The most important benefit of UPMGQ is also in complexity reduction. In particular, the complexities of UPMGQ with different solutions are comparable (with the same threshold $\theta$), which implies that UPMGQ is more capable of reducing the complexity for quantization with larger resolution (which can be observed in TABLE~\ref{tab:complexity}). The reason is that as long as enough bits are allocated to the higher levels group (i.e., enough bits to exceed the entropy of the higher levels group), increasing resolution is basically equivalent to adding levels to the lower levels group, which may not lead to significant complexity augment, since scalar quantization is performed for this group.

%

\subsection{Comparison of MSVQ and UPMGQ}
\label{sec:complexity_reduction}

In this subsection, we analyze and compare the complexities of MSVQ and UPMGQ.

MSVQ aims to reduce the complexity of searching operations, but is not capable of decrementing the codebook size. In fact, due to introducing extra stages, the total codebook size is even increased slightly for MSVQ. For example, for a two-stage MSVQ with $Q_{\textrm{MSVQ1}}$ and $Q_{\textrm{MSVQ2}}$ as the resolutions of each stage respectively, if performing vector length $L_{\text{MSVQ}}$ vector quantization on each of the stages, the complexities of searching operations (SO) and codebook size (CS) are given by
\begin{align}
&\textrm{SO}_{\textrm{MSVQ}}=2^{Q_{\textrm{MSVQ1}}\cdot L_{\textrm{MSVQ}}} + 2^{Q_{\textrm{MSVQ2}}\cdot L_{\textrm{MSVQ}}},\label{fun:SO_MSVQ}\\
&\textrm{CS}_{\textrm{MSVQ}}=2^{Q_{\textrm{MSVQ1}}\cdot L_{\textrm{MSVQ}}} + 2^{(Q_{\textrm{MSVQ1}}+Q_{\textrm{MSVQ2}})\cdot L_{\text{MSVQ}}}.\label{fun:CS_MSVQ}
\end{align}

On the other hand, UPMGQ tries to reduce the complexities of searching and storing at the same time. However, the reduction in searching complexity may not be as significant as that for MSVQ. More precisely, for UPMGQ with higher levels group using vector length $L_{\textrm{UPMGQ}}$ and resolution $Q_{\textrm{HIGH}}$ vector quantization, and with lower levels group using resolution $Q_{\textrm{LOW}}$ scalar quantization, the complexities of searching operations (SO) and codebook size (CS) are given by
\begin{align}
&\textrm{SO}_{\textrm{UPMGQ}}=2^{Q_{\textrm{HIGH}}\cdot L_{\textrm{UPMGQ}}} + 2^{Q_{\textrm{LOW}}},\label{fun:SO_UPMGQ}\\
&\textrm{CS}_{\textrm{UPMGQ}}=2^{Q_{\textrm{HIGH}}\cdot L_{\textrm{UPMGQ}}} + 2^{Q_{\textrm{LOW}}}.\label{fun:CS_UPMGQ}
\end{align}

\begin{table}[t!]
\caption{Complexity reduction of MSVQ and UPMGQ, without much degradation of EVM performance ($L_{\textrm{VQ}}=L_{\textrm{MSVQ}}=L_{\textrm{UPMGQ}}=2$). }
\centering
\begin{tabular}{l | l l l }
\toprule
					&\multicolumn{3}{c}{{\color{blue} EVM $(\%)$} / {\color{red} Codebook Size} /  {\color{magenta} Searching Operations} }																											 \\
\cmidrule{2-4}
					&\multicolumn{1}{c}{$Q=5$}											&\multicolumn{1}{c}{$Q=6$}											 &\multicolumn{1}{c}{$Q=7$} 												 \\
\midrule				
 $\;\;\;$VQ			&{\color{blue}$4.45$} / {\color{red}$1024$} / {\color{magenta}$1024$} 			 &{\color{blue}$2.43$} / {\color{red}$4096$} / {\color{magenta} $4096$ }			 &{\color{blue}$1.60$} / {\color{red} $16384$} / {\color{magenta} $16394$}				\\
$\;$MSVQ (two-stage)	&{\color{blue}$4.48$} / {\color{red}$1040$} / {\color{magenta} $\;\;80\;\;$}		 &{\color{blue}$2.45$} / {\color{red}$4160$} / {\color{magenta} $\;128\;$ }			 &{\color{blue}$1.62$} / {\color{red} $16448$} / {\color{magenta} $\;\;320\;\;$}			\\
UPMGQ ($\theta=0$)	 	&{\color{blue}$4.62$} / {\color{red}$\;264\;$} / {\color{magenta} $\;264\;$} 		 &{\color{blue}$2.52$} / {\color{red}$\;272\;$} / {\color{magenta} $\;272\;$} 		 &{\color{blue}$1.63$} / {\color{red}$\;\;288\;\;$} / {\color{magenta} $\;\;288\;\;$}		\\
UPMGQ ($\theta=-1$)		&{\color{blue}$4.59$} / {\color{red}$1028$} / {\color{magenta}$1028$ }			 &{\color{blue}$2.50$} / {\color{red}$1032$} / {\color{magenta} $1032$} 			 &{\color{blue}$1.62$} / {\color{red} $\;1040\;$}   / {\color{magenta} $\;1040\;$}			\\
\bottomrule
\end{tabular}
\label{tab:complexity}
\end{table}

TABLE~\ref{tab:complexity} illustrates the comparison of the search operation complexities, codebook size and the achieved EVM for VQ, MSVQ, and UPMGQ with vector length $2$, assuming an uplink AWGN channel, $5$dB SNR and $64$QAM modulation. Clearly, MSVQ and UPMGQ achieve significant complexity reduction, without significant impact to EVM performance.
Finally, we note that UPMGQ can also be incorporated with MSVQ to further reduce the searching complexity of the vector quantization procedure in $\bm{\mathsf{G2}}$.


\section{Simulation Results}
\label{sec:simulation}

In this section, link level simulation results are presented to illustrate the performance of the proposed vector quantization based CPRI compression algorithms.

\subsection{Downlink}
\label{sec:simulation:dl}

The simulation results for CPRI downlink are illustrated in Fig.~\ref{fig:downlink_simulaiton}. We assume CP removal block is enabled. We compare the EVM performances of SQ, VQ and MSVQ versus compression gain. Fig.~\ref{fig:downlink_simulaiton} shows that VQ with vector length $2$ and $3$ both perform better than SQ. MSVQ and UPMGQ perform close to regular vector quantization, but with reduced complexities as analyzed in the previous section. It is observed that $4.5$ times compression gain can be achieved with approximately $2\%$ EVM when using vector length $3$.

\begin{figure}[t!]
\hspace{5em}
\includegraphics[scale=0.6]{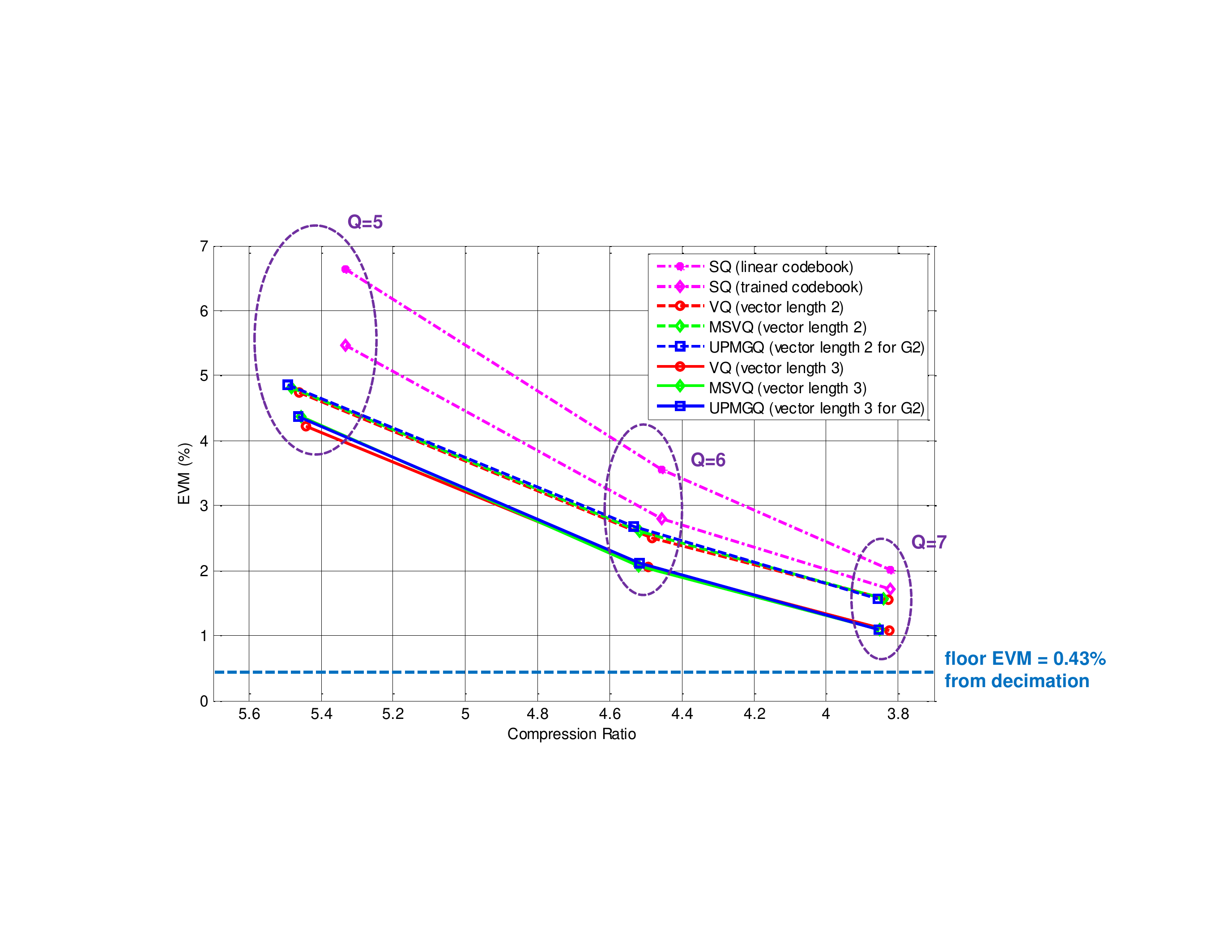}
\centering
\caption{Simulation results for downlink CPRI compression using different quantization methods.}
\label{fig:downlink_simulaiton}
\end{figure}

\subsection{Uplink}
\label{sec:simulation:ul}

Fig.~\ref{fig:uplink_simulaiton} presents the simulation results for CPRI uplink. SC-FDM signals with $64$QAM modulation are generated assuming AWGN channel model and $5$dB SNR. Similar to the downlink results, Fig.~\ref{fig:uplink_simulaiton} shows that VQ with vector length $2$ and $3$ both perform better than SQ. MSVQ and UPMGQ perform close to regular vector quantization. Our simulation results show that the proposed algorithms with vector length $3$ quantization can achieve $4$ times compression ratio with less than $2\%$ EVM.

\begin{figure}[t!]
\hspace{5em}
\includegraphics[scale=0.6]{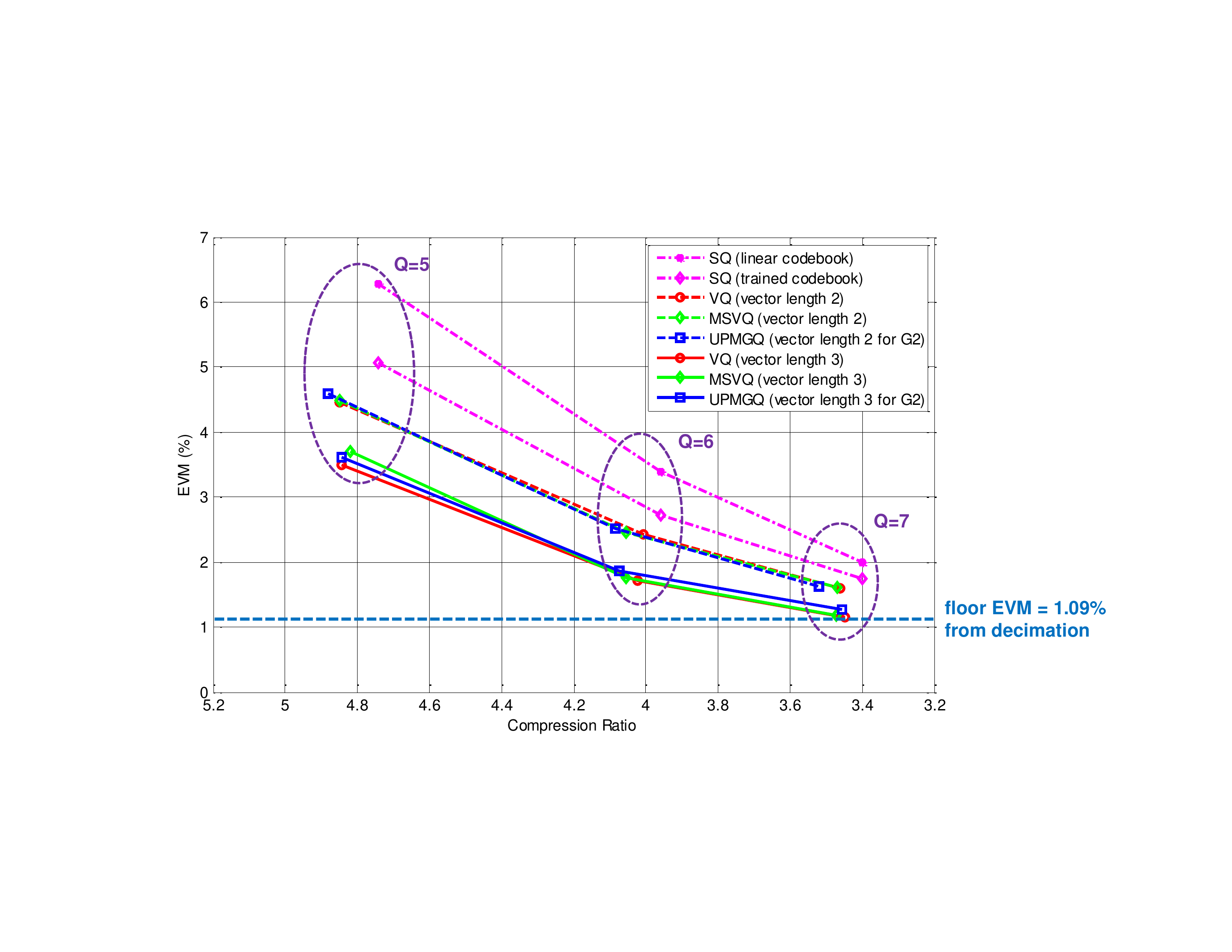}
\centering
\caption{Simulation results for uplink CPRI compression using different quantization methods.}
\label{fig:uplink_simulaiton}
\end{figure}

For vector quantization based CPRI compression to be practical for uplink, it is important for the scheme to work under different channel conditions, including fading, wide range of SNR and Doppler spread. To investigate the robustness of vector quantization, we perform VQ codebook training using training samples generated assuming AWGN channel and $64$QAM. The resultant codebook is then used to quantize different sets of CPRI uplink samples generated assuming 3GPP Ped B channel model, $16$QAM uplink data, for a wide range of SNRs and for different user speeds. The Block Error Rate (BLER) (with and without CPRI compression) versus SNR, is illustrated in Fig.~\ref{fig:bler} for $L_{\textrm{VQ}}=2$. Interestingly, the BLER performances with and without CPRI compression are virtually indistinguishable. The results show that the VQ codebook is remarkably robust against different channel conditions as well as data modulation scheme mismatches.

\begin{figure}[t!]
\includegraphics[scale=0.6]{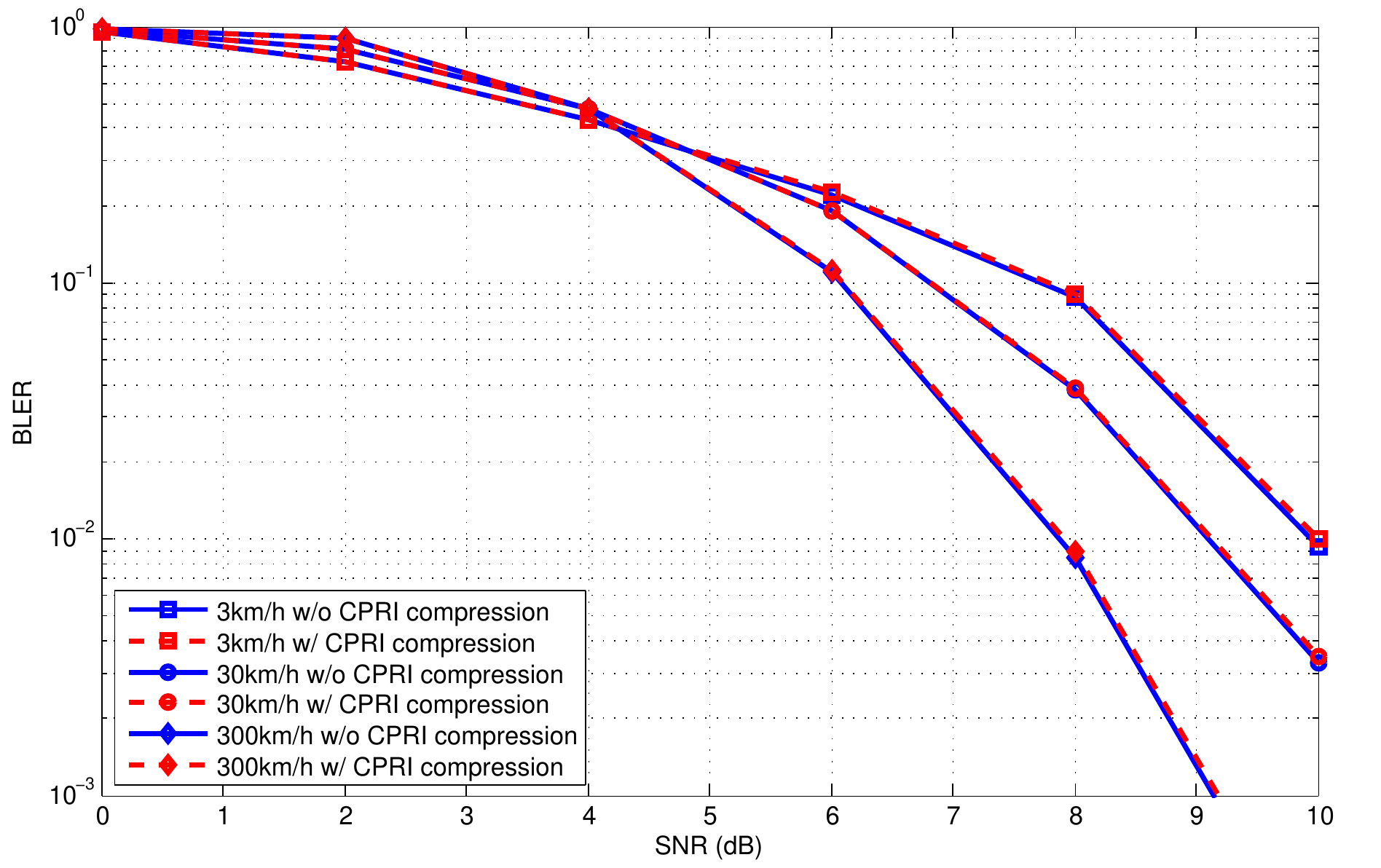}
\centering
\caption{BLER curves of uplink simulation with and without implementation of CPRI compression.}
\label{fig:bler}
\end{figure}


\section{Conclusion}
\label{sec:conclusion}

A vector quantization based CPRI compression framework is proposed in this paper. The vector quantizer codebook is trained using an enhanced Lloyd algorithm, designed to circumvent the persistent convergence to similar local optimum despite multiple Lloyd trials. To achieve complexity reduction without evident performance degradation, multi-stage vector quantization is investigated to significantly reduce the vector codebook search latency, and unequally protected multi-group quantization is proposed to remarkably reduce the codebook size. Comparing with scalar quantization in previous work, vector quantization based CPRI compression is shown to provide superior compression gain by exploiting the time correlation among the I/Q samples. Entropy coding over the resulting vector codebook can achieve additional compression through eliminating the potential redundancy from the distribution of codewords. Link level simulation results show that $4$ times compression for uplink and $4.5$ times compression for downlink can be achieved with approximate $2\%$ EVM distortion. Our simulation results also show remarkable robustness of the vector quantizer codebook against data modulation scheme mismatch, fading, wide range of SNR points and Doppler spread.


\appendices

\section{Details of Cyclic Prefix Removal, Decimation, and Block Scaling}
\label{app:details}

\subsection{Cyclic Prefix Removal}

In OFDM systems, cyclic prefix (CP) is prepended to IFFT output to create a guard period which helps eliminate inter-symbol interference from the previous symbol. However, CP represents a source of time domain redundancy as far as CPRI compression is concerned, and the \emph{CP Removal} block in CPRI compression module aims to eliminate this redundancy.

For instance, for a LTE $10$ MHz system with IFFT output length $L_{\textrm{SYM}}=1024$ and CP length $L_{\textrm{CP}}=128$, the compression gain from CP removal block is $\textrm{CR}_{\textrm{CPR}}=1.125$. We assume CPRI compression by CP removal can be performed for the downlink CPRI signals but not for the uplink CPRI signals because the uplink SC-FDM symbol timing is typically unknown at the RRU.

\subsection{Decimation}

In current LTE systems, the sampling rate exceeds the signal bandwidth, which results in redundancy in the frequency domain. For example, the sampling rate for an LTE $10$ MHz system is $15.36$ MHz, i.e. nearly one third of the bandwidth carries no information. The \emph{Decimation} block in CPRI compression module is introduced to eliminate this redundancy and can be implemented as a standard multi-rate filter \cite{oppenheim1989discrete} (see Fig.~\ref{fig:decimation}).  The input to \emph{Decimation} block is first $K$-times upsampled. The signal is then passed through a low-pass filter with cutoff frequency at $K\cdot f_s/L$, where $f_s$ is the original system sampling rate. Finally, the filtered signal is $L$-times downsampled ($K<L$). The sampling rate for output signal reduces to $K\cdot f_s/L$.

\begin{figure}[t!]
\includegraphics[scale=0.6]{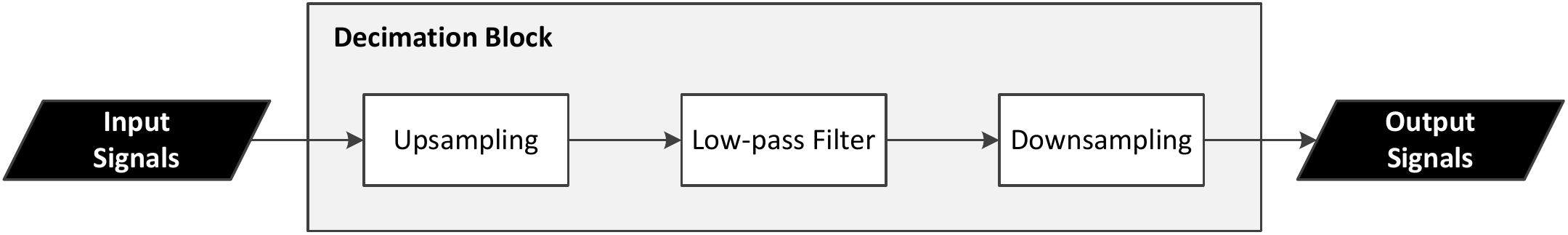}
\centering
\caption{Process chain in the \emph{Decimation} block for CPRI compression.}
\label{fig:decimation}
\end{figure}

However, decimation also causes signal distortion. Smaller $K/L$ provides larger compression gain, but leads to larger distortion as well. Simulation results illustrating the signal distortion as a result of decimation alone are shown in TABLE~\ref{tab:decimation}.  In this simulation, samples are generated from LTE $10$ MHz system uplink with AWGN channel, $5$dB SNR, and $64$QAM modulation. The measurement of distortion is frequency domain EVM defined by \eqref{fun:FD_EVM}. A threshold effect is evident from the table: when decimation value is smaller than around $0.6$, EVM increases remarkably. The EVM from decimation provides a floor distortion for downstream blocks (see Fig.~\ref{fig:downlink_simulaiton} and Fig.~\ref{fig:uplink_simulaiton}).

\begin{table}[t!]
\caption{Effect of different choices of decimation values.}
\centering
\begin{tabular}{c c c c c | c c c}
\toprule
$K/L$	&$15/16$	&$3/4$ 	&$2/3$	&$5/8$	&$15/28$	&$15/32$	&$5/12$\\
\midrule
EVM ($\%$)	&$0.23$	&$0.61$	&$0.96$	&$1.09$	&$26.14$	&$46.19$	&$52.72$\\
\bottomrule
\end{tabular}
\label{tab:decimation}
\end{table}

\subsection{Block Scaling}

Signal on CPRI link can have a large dynamic range. For instance, for uplink CPRI signals, different propagation path loss, shadowing, fading channel and mobility for different users may result in significant variance in general. \emph{Block Scaling}, also known as automatic gain control (AGC), is employed to lower the resolution of signal and to maintain the dynamic range simultaneously \cite{samardzija2012compressed}\cite{hogenauer1981economical}.

In AGC, a sequence of signal samples are grouped into blocks of consecutive samples, where the samples in each block are scaled by a scaling factor which can vary from block to block.  More precisely, assume $N_\textrm{BS}$ number of samples form a block, and the largest absolute value of I/Q components for a particular block $b\in\{1,\ldots,M/N_\textrm{BS}\}$ is determined by
\begin{align}
A(b)=\max_{m=N_{\textrm{BS}}\cdot (b-1)+1,\ldots,N_{\textrm{BS}}\cdot b}\left\{|\textfrak{R}(\bm{s}_{\textrm{DEC}}(m))|,|\textfrak{I}(\bm{s}_{\textrm{DEC}}(m))|\right\},\nonumber
\end{align}
where $\bm{s}_{\textrm{DEC}}(m)$ is the output from the \emph{Decimation} block, and $\textfrak{R}(\cdot)$ and $\textfrak{I}(\cdot)$ represent the real and imaginary parts of a complex-valued sample, respectively.
Next, the corresponding scaling factor for block $b$ is determined by
\begin{align}
S(b)=\left\{
\begin{array}{ll}
\lceil A(b) \rceil,	& \text{for }\lceil A(b) \rceil\leq 2^{Q_{\textrm{BS}}}-1,\nonumber\\
2^{Q_{\textrm{BS}}}-1,	& \text{for }\lceil A(b) \rceil > 2^{Q_{\textrm{BS}}}-1. \nonumber
\end{array}
\right.\nonumber
\end{align}
By this definition, $S(b)$ is an integer that can be represented with $Q_{\textrm{BS}}$ bits. Samples in block $b$ are then scaled to produce an output $\bm{s}_{\textrm{BS}}(m)$ as follows.
\begin{align}
\bm{s}_{\textrm{BS}}(m)=\bm{s}_{\textrm{DEC}}(m)\cdot \frac{2^{Q_{\textrm{VQ}}}-1}{S(b)}, \quad N_{\textrm{BS}}\cdot (b-1)+1 \leq m\leq N_{\textrm{BS}}\cdot b,\nonumber
\end{align}
where $Q_{\textrm{VQ}}$ denotes the number of bits per complex component for vector quantization. In essence, block scaling normalizes the input signals for vector quantization.

There is a slight increase in signal processing latency associated with block scaling. For instance, in LTE $10$ MHz system, by choosing $N_{\textrm{BS}}=32$, the latency due to block scaling is $3.33$ $\mu$s with decimation value chosen as $5/8$.




\bibliographystyle{IEEEtran}


\end{document}